\documentclass{llncs}

\usepackage[utf8]{inputenc}
\usepackage{url}

\usepackage{stmaryrd}

\usepackage{amsmath}

\usepackage{mathtools}

\def\arrayin#1{\begin{array}{rcl}#1\end{array}}

\newcommand{\DD}{\mathit{\mathbb{I}}}

\newcommand{\DL}{\mathcal{GDL}(\A)}
\newcommand{\CDL}{\mathcal{CGDL}(\A)}
\newcommand{\FDL}{\Gamma}

\newcommand{\bool}{\mathbf{2}}

\newcommand{\cmodels}{\models_{\mathcal{CGDL}}}

\def\M{\mathbb{M}}

\newcommand{\Prop}{\mathrm{Prop}}
\newcommand{\Prog}{\mathrm{Prg}}
\newcommand{\Fm}{\mathrm{Fm}^{\FDL(\A)}}

\def\A{\mathbf{A}}
\def\AA{\mathcal{A}}

\newcommand{\Mod}{\mathrm{Mod}}

\def\prgint#1{\llbracket#1\rrbracket}

\newcommand{\REL}{\REL}

\def\square{\Box}

\def\PL{\mathit{PL}}
\def\REL{\mathit{REL}}

\def\H2PL{\mathcal{H}^2\PL}

\def\Lat{\mathbf{L}}
\def\un{\mathbf{1}}

\newcommand{\0}{\mathbf{0}}

\def\rcb#1#2#3#4{\def\nothing{}\def\range{#3}\mathopen{\langle}#1 \ #2 \ \ifx\range\nothing::\else: \ #3 :\fi \ #4\mathclose{\rangle}}

 \def\just#1#2{\\ &#1& \rule{2em}{0pt} \{
\mbox{\rule[-.7em]{0pt}{1.8em} \footnotesize{#2}} \} \\ && }

\def\justdois#1#2#3{\\ &#1& \rule{2em}{0pt} \{
  \mbox{\rule[-.7em]{0pt}{1.8em} \footnotesize{#2}} \\ &&
  \mbox{\rule[-.7em]{0pt}{1.8em} \footnotesize{#3}} \} \\ && 
  }

\DeclareMathAlphabet{\mathbb}{U}{msb}{m}{n}
\DeclareSymbolFont{ams}{U}{msa}{m}{n}
\DeclareSymbolFontAlphabet{\mathams}{ams}
\DeclareMathSymbol{\filter}{\mathams}{ams}{22}

\usepackage{atbegshi}% http://ctan.org/pkg/atbegshi
\AtBeginDocument{\AtBeginShipoutNext{\AtBeginShipoutDiscard}}

\usepackage{xypic}
\usepackage{tikz}
\usepackage{courier}

\usepackage{latexsym}

\usetikzlibrary{calc,trees,positioning,arrows,chains,shapes.geometric,%
    decorations.pathreplacing,decorations.pathmorphing,shapes,%
    matrix,shapes.symbols,shapes.arrows,automata}

\tikzstyle{obj} =[circle, minimum width=0.5cm, minimum height=0.5cm, draw=black]

\tikzset{
>=stealth',
  punktchain/.style={
    rectangle, 
    rounded corners, 
    fill=blue!20,
    draw=black, very thick,
    text width=10em, 
    minimum height=5em, 
    text centered, 
    on chain},
      punktchain2/.style={
        rectangle, 
        fill=yellow!20,
        draw=black, thick,
        text width=10em, 
        minimum height=2em, 
        text centered, 
        on chain},
         punktchain3/.style={
                rectangle,
                rounded corners,
                fill=green!20,
                draw=black, thick,
                text width=2em, 
                minimum height=2em, 
                text centered, 
                on chain},
                 punktchain4/.style={
                        rectangle, 
                        fill=yellow!20,
                        draw=black, thick,
                        text width=2em, 
                        minimum height=2em, 
                        text centered, 
                        on chain},
                         punktchain5/.style={
                                                rectangle,
                                                text width=6em, 
                                                minimum height=2em, 
                                                text centered
                                                },
  line/.style={draw, thick, <-},
  element/.style={
    tape,
    top color=white,
    bottom color=blue!50!black!60!,
    minimum width=8em,
    draw=blue!40!black!90, very thick,
    text width=10em, 
    minimum height=3.5em, 
    text centered, 
    on chain},
  every join/.style={->, thick,shorten >=1pt},
  decoration={brace},
  tuborg/.style={decorate},
  tubnode/.style={midway, right=2pt},
}

\begin{document}

\begin{frontmatter}

\title{On the construction of multi-valued concurrent dynamic logic}

\thanks{\tiny This work was founded by the ERDF --- European Regional Development Fund through the Operational Programme for Competitiveness and Internationalisation - COMPETE 2020 Programme and by National Funds through the Portuguese funding agency, FCT - Funda\c{c}\~ao para a Ci\^encia e a Tecnologia, within project \texttt{POCI-01-0145-FEDER-030947}.}

\author{Leandro Gomes}
	\institute{ HASLab INESC TEC - Univ. Minho, Portugal}

\end{frontmatter}

\maketitle

\begin{abstract}
Dynamic logic is a powerful framework for reasoning about imperative programs. An extension with a concurrent operator \cite{ConcurrentDL} was introduced to formalise programs running in parallel. In other direction, other authors proposed a systematic method for generating multi-valued propositional dynamic logics to reason about weighted programs \cite{Madeira2016}. This paper presents the first step of combining these two frameworks to introduce uncertainty in concurrent computations. In the developed framework, a weight is assigned to each branch of the parellel execution, resulting in a (possible) asymmetric parallelism, inherent to fuzzy programming paradigm \cite{DBLP:journals/artmed/VetterleinMA10,DBLP:journals/ijcisys/CingolaniA13}. By adopting such an approach, a family of logics is obtained, called \emph{multi-valued concurrent propositional dynamic logics} ($\CDL$), parametric on an action lattice $\A$ specifying a notion of ``weight'' assigned to program execution. Additionally, the validity of some axioms of CPDL is discussed in the new family of generated logics.

\end{abstract}

\section{Introduction}\label{sec:intro}

%!TEX root = main.tex

Over time, the different variations of dynamic logics developed went hand-in-hand with the very notion of its object, the \emph{program}. This resulted in a diverse myriad of dynamic logics tailored to specific programming paradigms. Examples include probabilistic \cite{Kozen1985}, concurrent \cite{ConcurrentDL}, quantum \cite{Baltag2012} and continuous \cite{Platzer} computations, and combinations thereof.
An example of another non-trivial paradigm is the fuzzy one \cite{DBLP:journals/artmed/VetterleinMA10,DBLP:journals/ijcisys/CingolaniA13}, where the execution of a program differs from both classical and probabilistic scenarios: a conditional statement may act as a concurrent execution with a weight associated to each branch. The formalisation of such behaviour encompasses two non-trivial computational settings: concurrency and uncertainty. An extensive research can be found in the literature on diverse formalisms to reason about programs running in parallel \cite{DBLP:conf/concur/HoareMSW09,DBLP:journals/jlp/JipsenM16} and to deal with uncertainty \cite{Kozen1985,DenHartog2002,Qiao2008,probProgUsingHoareL}. However, even when these two components are combined into a single framework \cite{DBLP:journals/corr/McIverRS13}, the uncertainty %is introduced in the non-deterministic choice operator,
models probabilistic nondeterminism. Thus we are still missing a proper semantics to describe the behaviour of the fuzzy paradigm.

Recently, reference \cite{Madeira2016} initiated a research agenda on the systematic development of multi-valued propositional dynamic logics, parametric on an action lattice, which defines both the computational paradigm where programs live, and the truth space where assertions are evaluated. %Later, a more algebraic approach \cite{DBLP:conf/sbmf/GomesMB17,SCP} led to generalisations of Kleene algebra with tests (KAT), also to give semantics to weighted computations.
%Other research line was concerned about how to reason about programs running in parallel. On the logical side,
Following another research line, an extension to propositional dynamic logic (PDL) was introduced in reference \cite{ConcurrentDL}, called \emph{concurrent propositional dynamic logic (CPDL)}, to reason about concurrent computations. In the models presented for this logic, the programs are interpreted as \emph{binary multirelations}, to describe a parallel execution from a state to a set of states.

Combining these research lines, this paper takes the first step on the development of a method to generate multi-valued concurrent propositional dynamic logics. As in \cite{Madeira2016}, the logics are parametric on a generic action lattice, to model both the computational domain and a (possible graded) truth space where the assertions about programs are evaluated.
First, the semantics of CPDL is adapted to model programs as weighted parallel executions, by introducing the concept of \emph{fuzzy multirelations}. That means that a program, in the new logics, is interpreted as a relation between a state and a fuzzy set of states. The intuition is that the weights of the fuzzy set may describe an execution probability of each branch of the program, an asymmetric parallel flow or even the energy/costs associated to each branch.
The second step of this paper consists on presenting the actual method of generating (parametric) multi-valued CPDL. %using a method analogous to \cite{Madeira}.
The family of the resulting logics is called $\CDL$.

%In general terms, the framework applies an analogous method of \cite{Madeira2016} to CPDL. In this scenario, programs are interpreted as fuzzy multirelations, where the weights associated with  intuitively describing a (possible asymmetric) parallel execution.

%This paper takes the first step on generalising CPDL, by adopting some ideas from the previous work mentioned above. As a starting point, and following \cite{ConcurrentDL}, we characterise the semantics where the programs for this logic are interpreted, by introducing the notion of \emph{fuzzy multirelation}; %such semantics constitutes the computational model where the programs of the logic are interpreted, in an analogy to \cite{Madeira2016}.
%Then, a (parametric) dynamic logic is induced by a complete action lattice $\A$, using a method analogous to \cite{Madeira}. The family of the generated logics is called $\CDL$.

% Built on these motivations, we resort to the systematic method of generating multi-valued PDL \cite{Madeira2016}, parametric on an action lattice. Such algebraic structure serves the purpose of giving semantics to both the space of computations, and a proper (multi-valued) valuation of the logic formulas. To model the space of computations where the programs are interpreted, we present a generalisation of the concept of binary multi relations, caleed \emph{fuzzy multirelations}.

This paper is organised as follows. Section \ref{sec:prelim} presents a brief background overview. Then, Section \ref{sec:mvcdl} starts to introduce fuzzy multirelations and defines some operations over them. Such algebra is the mathematical formalism in which programs are interpreted in the generated logics. The same section ends with the study of an axiomatisation for the generated logics. Finally, Section \ref{sec:conc} concludes and enumerates topics for future work.

\section{Preliminaries}\label{sec:prelim}

%!TEX root = main.tex

\subsection{Semantics for concurrency}

The semantics of CPDL is based on the concept of \emph{binary multirelation}. The relevant definition and some operators are recalled below.

	\begin{definition}[Binary multirelation \cite{DBLP:journals/tocl/FurusawaS16}]\label{def:binmrel}
	Given a set $X$, a binary multirelation is a subset of the cartesian product $X\times P(X)$, i.e. a set of ordered pairs $(a,A)$, where $a\in X$ and $A\subseteq X$. The following operations over multirelations are defined:

\begin{itemize}
	\item $R\cup S$ as the union of $R$ and $S$;
	\item the \emph{Peleg sequential composition}

$$R\cdot S=\Big\{(a,A)\mid \exists B.(a,B)\in R\wedge \exists f.(\forall b\in B. (b,f(b))\in S)\wedge A=\bigcup f(B)\Big\};$$
\item the \emph{parallel composition} % \cite{DBLP:journals/tocl/FurusawaS16}
$R\cap S=\{(a,A\cup B)\mid (a,A)\in R \wedge (a,B)\in S\}$.
\end{itemize}

\end{definition}

% \begin{definition}[Angelic and Demonic choice \cite{DBLP:journals/scp/MartinCR07}]\label{def:angelic_demonic}
% The \emph{angelic choice} of binary multirelations $R$ and $S$ is the union $R\cup S$.\\
% The \emph{demonic choice} of binary multirelations $R$ and $S$ is the intersection $R\cap S$.
% \end{definition}

%The formalism used in this paper resorts to Peleg's work on Concurrent dynamic logic \cite{ConcurrentDL}. The operations of sequential and parallel compositions were defined as \cite{DBLP:journals/tocl/FurusawaS16}:

Note that the union of binary multirelations is just the set union. %, as defined for modelling angelic nondeterminism.
The sequential composition operator is rather more complex. A pair $(a,A)$ belongs to the sequential composition of multirelations $R$ and $S$ if and only if $a$ is related with some intermediate set of states $B$ and every $b\in B$ must be related with some subset of $A$ such that the union of all those subsets is $A$. Finally, an element $(a,A)\in R\cap S$ indicates a parallel execution of a program from a state $a$ to a set of states in $A$, ``combining'' the arriving states of $R$ and $S$ into $A$. Note that such composition is dual to $R\cup S$, where $(a,B)$ and $(a,C)$ correspond to distinct executions. The first kind of choice in commonly called \emph{demonic}, while the latter is known as \emph{angelic}.

% Relatively to sequential composition we assume, in this work, an alternative definition, the \emph{Kleisli composition} \cite{DBLP:journals/jlp/FurusawaKST17}:

% \begin{equation}R\circ S=\Big\{(a;A)\mid \exists B\cdot(a,B)\in R\wedge A=\bigcup S(B) \Big\}\label{eq:kleisli}\end{equation}

% \noindent where $S(B)=\{C\in P(Z)\mid \exists b\in B\cdot(b,C)\in S\}$.

\subsection{Concurrent propositional dynamic logic}

Concurrent propositional dynamic logic (CPDL), as introduced in \cite{ConcurrentDL}, is an extension of PDL with a parallel operator $\cap$ added to the syntax of programs.% Hence, both its sets of programs and formulas are that of PDL plus the ones containing expressions generated by $\pi\cap \pi$.
 The semantics interprets programs as binary multirelations $R\subseteq W\times P(W)$, where composed programs are interpreted according to the operators of Definition \ref{def:binmrel}. Intuitively, an element $(a,A)$ of a binary multirelation expresses that the a program executed from a state $a$ ends in all states of $A$ in parallel. Models of CPDL consist of tuples $(W,V,\prgint{-})$ where $W$ is a set of states, $V$ is a valuation function which attributes a subset of $W$ to each atomic formula, and $\prgint{-}$ attributes a subset of $W\times P(W)$ to each atomic program. For instance, the formula $\langle\pi\rangle\rho$ holds in a state $w$ if and only if $\exists U\subseteq W$ s.t. $(s,U)\in \prgint{\pi}$ and $U\in V(\rho)$. For more details about the semantics of CPDL see \cite{ConcurrentDL}.
The axiom system of CPDL is that of PDL with the additional axiom
$\langle\pi_1\cap\pi_2\rangle \rho\equiv \langle\pi_1\rangle\rho\wedge \langle\pi_2\rangle\rho$ and restricting
$\langle\pi_0\rangle(\rho\vee\rho')\equiv \langle\pi_0\rangle \vee \langle\pi_0\rangle\rho'$ to atomic programs.

\subsection{Parametric construction of multi-valued dynamic logics}

Thus subsection provides a short review of the dynamisation method introduced in \cite{Madeira2016}. Let us start by revisiting the following definition:

 \begin{definition}[\cite{onactionalgebras}]\label{def:actionlattice}
 An \emph{action lattice} is a tuple
	%\[\A=(A,+,;,0,1,*,\leftarrow,\rightarrow,\cdot)\]
	$\A=(A,+,;,0,1,*,\rightarrow,\cdot)$, that is a residuated lattice with order $\leq$ induced by $+$: $a\leq b$ iff $a+b=b$, plus the axioms
	% \begin{align}
		$1+a+(a^*;a^*) \leq a^*$ and %\label{eqn10}\\
	 	$(x\rightarrow x)^* = x\rightarrow x$. %\label{eqn15}
\end{definition}

An action lattice is called a \emph{$\DD$-action lattice} when the identity of the $;$ operator %(axiom \ref{eqn6})
coincides with the greatest element of the residuated lattice, i.e. $1=\top$. Moreover, an action lattice $\A$ is complete when every subset of $\A$ has both supremum and infimum. Since operators $+$ and $;$ are associative, we can generalise them
to $n$-ary operators and use the notation $\sum$ and $\prod$ to represent their iterated
versions, respectively. The generation of dynamic logics illustrated in the Section \ref{sec:mvcdl} will be parametric on the class of complete action lattices, since completeness is required to ensure the existence of arbitrary suprema. The general construction of multi-valued dynamic logics is revisited bellow.\\

% Below we recall some properties of action lattices from \cite{Madeira2016}.

% \begin{lemma}\label{lem:lemma1}
% The following properties hold for any action lattice $\A$:

% \begin{eqnarray}
% a \leq b \; \& \; c\leq d & \Rightarrow & a+c \leq b+d \label{eq:eq33} \\
% 	a;(b\cdot c) & \leq & (a;b)\cdot (a;c)	 \label{eq:eq36}
% \end{eqnarray}

% \noindent For $I$ finite, we also have

% \begin{equation}
% 	\sum_{i\in I}(a_i\cdot b_i) \leq  \sum_{i\in I} a_i \cdot \sum_{i\in I} b_i
% \label{eq:eq45}
% \end{equation}

% \end{lemma}

% \begin{lemma}\label{lemma:prova0}
% 	The following properties hold in any $\DD$-lattice:
% 	\begin{enumerate}
% 		\item \label{lem:2.5}$1=a;b \Leftrightarrow a=1 \; \&\;  b=1$ 
%  		\item\label{lem:2.6} $a\leq b \Leftrightarrow (a\rightarrow b)=1$
%  	\end{enumerate}
% \end{lemma}

\noindent \textbf{Signatures.} Signatures of $\DL$ are pairs $(\Pi,\Prop)$ corresponding to the denotations of atomic programs and  propositions, respectively.

\medskip
\noindent \textbf{Formul\ae.}
The \emph{set of composed programs}, denoted by $\Prog(\Pi)$, contains all expressions generated by 
	$\pi \ni \pi_0 \,|\, \pi;\pi \,|\, \pi + \pi \,|\, \pi^*$
	for $\pi_0\in \Pi$. 
	Given a signature $(\Pi,\Prop)$, the $\DL$-formul\ae\, for $(\Pi,\Prop)$ are the ones generated by the grammar
$\rho \ni \top\,|\, \bot\,|\, p \,|\, \rho \vee \rho \,|\,\rho \wedge \rho \,|\, \rho \rightarrow \rho \,|\, \rho \leftrightarrow \rho \,|\,\langle \pi \rangle\rho\,|\,[\pi]\rho$
for $p\in \Prop$ and $\pi \in \Prog(\Pi)$. 
%We use the double directional arrow $\rho \rightleftarrow \rho'$ to denote  formulas $\rho \rightarrow \rho' ; \rho' \rightarrow \rho$
%Note that this corresponds to the \emph{positive} fragment of the propositional dynamic logic.
  
 \medskip\noindent \textbf{Semantics.} The space where the computations of $\DL$ are interpreted is given by the algebra
$\M_n(\A)=(M_n(\A),\textbf{+},\textbf{;},\mathbf{0},\mathbf{1},\textbf{*})$
\noindent where $M_n(\A)$ is the space of $(n \times n)$-matrices over $\A$, the operators $+$, $;$ are the usual matrix sum and multiplication, respectively, $\0$, $\un$ are the zero matrix and the identity matrix, respectively, and $\textbf{*}$ is the operator defined as in \cite{conway1971,kleenematrices}. The matrix representation of a program expresses, for each pair of states $s$, $s'$, the weight (e.g. probability, cost, uncertainty) of the program going from $s$ to $s'$.

% \begin{enumerate}
% 	\item $M_n(\A)$ is the space of $(n \times n)$-matrices over $\A$
% 	\item  for any $A,B \in M_n(\A)$, define $M=A \mathbf{+} B$ by $M_{i,j}=A_{i,j} + B_{i,j}$, $i,j \leq n$.
% 	\item for any $A,B \in M_n(\A)$, define $M=A\;\textbf{;}\;B$ by  $ M_{i,j}= \sum_{k=1}^{n}(A_{i,k};B_{k,j})$ for any $i,j\leq n$.
% 	\item $\mathbf{1}$ and $\mathbf{0}$ are the $(n \times n)$-matrices defined by ${\mathbf{1}}_{i,j}= \begin{cases}1 & \mbox{ if } i=j \\
% 	0 & \mbox{otherwise}\end{cases}$ and $\mathbf{\mathbf{0}}_{i,j}=0$, for any $i,j\leq n$.

% 	\item for any $M=[a]\in \M_1(\A)$, $M^\textbf{*}=[a^*]$; 

% 	for any ${M=\left[
%       \begin{array}{l|r}
%         A & B\\
%         \hline
%         C & D
%       \end{array}
%     \right]}\in M_n(\A)$, $n>1$, where $A$ and $D$ are square matrices, define \[M^\textbf{*}=\left[
%       \begin{array}{l|r}
%         F^{\textbf{*}} & F^{\textbf{*}}\,\textbf{;}\,B\,\textbf{;}\,D^{\textbf{*}}\\
%         \hline
%         D^*;C;F^* & D^{\textbf{*}} \textbf{+}(D^{\textbf{*}}\,\textbf{;}\,C\,\textbf{;}\,F^{\textbf{*}}\,\textbf{;}\,B\,\textbf{;}\,D^{\textbf{*}})
%       \end{array}
%     \right]\] where $F=A+B\,\textbf{;}\,D^{\textbf{*}}\,\textbf{;}\,C$. Note that this construction is recursively defined from the base case (where $n=2$) where the operations of the base action lattice $\A$ are used.
 
% \end{enumerate}

  $\DL$-models for a signature ($\Prop$,$\Pi$), denoted by $\Mod^{\DL}(\Pi,\Prop)$, consists of tuples	
		 $\AA=(W,V,(\AA_\pi)_{\pi\in \Pi})$ where $W$ is a finite set (of states), $V:\Prop \times W \rightarrow A$ is a valuation function, and  $\AA_\pi \in M_n(\A)$, with  $n$ standing for the cardinality of $W$.

The interpretation of a program $\pi\in \Prog(\Pi)$ in a model \\$\AA \in \Mod^{\DL}(\Pi,\Prop)$ is recursively defined, from the set of atomic programs $(\AA_\pi)_{\pi\in \Pi}$, as $\AA_{\pi;\pi'}=\AA_\pi \,\textbf{;} \,\AA_{\pi'},
	 \AA_{\pi + \pi'}=\AA_\pi \,\textbf{+}\, \AA_{\pi'}\; \text{and}\; 
	\AA_{\pi^*}=\AA_{\pi}^{\textbf{*}}$.\\

	\noindent \textbf{Satisfaction.} The (graded) satisfaction relation, for a model\\ $\AA\in
\Mod^{\DL}(\Pi,\Prop)$, with $\A$ complete, consists of a function\\ $\models \; : W \times \Fm(\Pi,\Prop) \rightarrow A$
recursively defined as follows: 
	\begin{itemize}
		\item $(w\models \top) = \top$
		\item $(w\models \bot) = \bot$
		\item $(w\models p) = V(p,w)$, for any $p\in \Prop$
		\item $(w\models \rho \wedge \rho') = (w\models \rho) \cdot (w\models \rho')$
		\item $(w\models \rho \vee \rho') = (w\models \rho) + (w\models \rho')$
		\item $(w\models \rho \rightarrow \rho') = (w\models \rho) \rightarrow (w\models \rho')$
		\item $(w\models \rho \leftrightarrow \rho') = (w\models \rho \rightarrow \rho');(w\models \rho' \rightarrow \rho)$
		\item $(w\models \langle \pi \rangle\rho) = \sum\limits_{w'\in W}\big(\AA_{\pi}(w,w');(w'\models \rho)\big)$
		\item $(w\models [\pi] \rho) = \prod\limits_{w'\in W}\big(\AA_{\pi}(w,w')\rightarrow(w'\models \rho)\big)$
	\end{itemize}
 %We say that $\rho$ is \emph{valid} when, for any model $\AA$, and for each state $w\in W$, $(w\models\rho) =\top$.

 The (graded) satisfaction in a given state gives the degree of certainty of a formula in such state. For instance $M,w\models\langle\pi\rangle\rho$ gives the certainty that $\rho$ is achieved from state $w$ through the execution of $\pi$.
It is relevant to note that $\DL$ is a generalisation of PDL, for each action lattice $\A$. In particular, by considering the Boolean lattice, the generated logic $\mathcal{GDL}(\mathbf{2})$ coincides with PDL.% In this case, the ``weights'' associated to the execution of a program indicate the presence (or absence) of a transition between states. Relatively to the satisfaction relation, the possible two values ($\top$ and $\bot$) represent the classic ``true'' and ``false''.

\section{Multi-valued concurrent dynamic logic}\label{sec:mvcdl}

%!TEX root = main.tex

% To characterise the space of computations where programs are to be interpreted in this logic, we

% This subsection introduces the semantics where programs, in CPDL, are to be interpreted. Such is formalise dby the notion of \emph{fuzzy binary multirelations}.

Before presenting the construction of the logic, we introduce the mathematical formalism to define the model where the programs will be interpreted.

\subsection{Fuzzy binary multirelations}

\begin{definition}[Fuzzy set \cite{ZADEH1965338}]\label{def:fuzzyset}
	Given a set $X$ and a complete residuated lattice $\Lat$, a \emph{fuzzy subset} of $X$ is a function $\phi :X\to L$; $\phi(x)$ defines the \emph{membership degree of $x$ in $\phi$}. The set of all fuzzy subsets of $X$ is denote as $L^X$.
% \end{definition}
% \begin{definition}[Fuzzy relation \cite{ZADEH1965338}]\label{def:fuzzyrel}
	% Let $X_1,X_2,...,X_n$ be sets. A \emph{fuzzy relation} $\mu$ between $X_1,X_2,...,X_n$ is a fuzzy subset of the cartesian product $X_1\times X_2\times ...X_n$.
The support of $\phi$ is a fuzzy subset $\psi$ such that $\psi(x)>0$, $\forall x\in X$.

\end{definition}

% For each $x_1\in X_1,x_2\in X_2,..., x_n\in X_n$, $\mu(x_1,x_2,...,x_n)$ can be interpreted as the truth value of how elements $x_1,x_2,...,x_n$ are related by $\mu$. Therefore, as fuzzy sets model collections of objects, where there is a degree to which an object belongs to a collection, fuzzy relations model relationships between objects up to some membership degree. For the purpose of this work, we consider only binary fuzzy relations. %So, every time we mention the term \emph{fuzzy relation}, we are referring to fuzzy subsets of the cartesian product $X_1\times X_2$ .

% \begin{definition}[Binary multirelation \cite{DBLP:conf/RelMiCS/Rewitzky03}]\label{def:binmrel}
% 	Given a set $X$, a binary multirelation is a subset of the cartesian product $X\times P(X)$, i.e. a set of ordered pairs $(a,A)$, where $a\in X$ and $A\subseteq X$.
% \end{definition}

Since an action lattice is an extension of a residuated lattice, the concept of fuzzy set can be defined as well for the former. Such is the case for all the remaining formalisms introduced in this paper.

\begin{definition}[Fuzzy binary multirelation]\label{def:fuzzymrel}
	Given a set $X$ and a complete action lattice $\A$ over carrier $A$, a \emph{fuzzy binary multirelation} $R$ over $X$ is a set  $R\subseteq X\times A^X$. The following operations for fuzzy binary multirelations are defined:
	
	\begin{itemize}
		\item $R\cup S$ as the union of $R$ and $S$;
		%\item $R\cdot S=\{(a,\phi)\}$, with $\phi(a)=\sum_{x\mid(x,\phi_i)\in S}\phi(x);\phi_i(a)$, $i\in \mathbb{N}$;
		%\item $R\cdot S=\{(a,\phi)\mid\exists B.(a,B)\in R\mbox{ and } \forall b\in B. (b,A)\in S\}$, called the composition of multirelations;
		 \item %the \emph{sequential composition}
		 $R\cdot S =\Big\{(a,\phi)\mid %(a,\phi_a)\in R
         \phi(c)=\sum_{(a,\phi_a)\in R}\Big(\prod_{(b,\phi_b)\in S} \phi_a(b);\phi_b(c) \Big)\Big\}$
		\item $R \cap S = \{(a,\phi_R\cup\phi_S)\mid (a,\phi_R)\in R \text{ and } (a,\phi_S)\in S\}$, where $\phi_R\cup\phi_S$ is the union of fuzzy sets $\phi_R$ and $\phi_s$, as defined in \cite{ZADEH1965338};
		\item $R^*=\bigcup\{R^n:n\geq 0\}$.
	\end{itemize}

\end{definition}
We denote by $M(X)$ the set of all fuzzy binary multirelations over $X$.

Note, particularly, how this definition generalises the concept of binary multirelations, particularly to the role of lattice $\A$. This structure supports a set of truth values beyond the classical true and false, which are associated to the elements of the second component of $R$. %This way they represent ``internal'' truth values, while the fuzzy multirelation itself is an ordinary subset.% Another way would be to consider a fuzzy multirelation as being a fuzzy subset of $X\times P(X)$, or even a fuzzy subset of $X\times A^X$, the latter with two ``levels of fuzziness''. Nevertheless Definition \ref{def:fuzzymrel} is the most suitable for the purpose of this work.
By using such formalisation we are able to model a program as an execution with multiple ``arrows'' leaving a state to a set of states in parallel, with a (possible different) fuzziness degree associated with each ``arrow''.% Such notion would describe a (possible) assymetric notion of parallelism, where each program is the representation of multiple (possible different) flows of an execution.
Note that if $\A$ is the Boolean lattice $\mathbf{2}$, any fuzzy multirelation $R\subseteq X\times \mathbf{2}^X$ is a binary multirelation. Since the goal is still to model programs as binary input-output relations, only the binary case is considered, and thus the remaining of this paper refers to fuzzy binary multirelations simply as fuzzy multirelations. Another aspect that is relevant for the formalisation of the logics is the restriction to fuzzy multirelations $R\subseteq X\times A^X$ where the fuzzy set $\phi$ in $A^X$ is defined such that $\phi(x)>0$, $\forall x\in X$. In other words, we take only the support of fuzzy sets for the fuzzy multirelations considered in this paper.

The operations for fuzzy multirelations are interpreted buying intuitions from the classic definition. One such case is the operator $\cup$, which corresponds to the classical set union. %The sequential composition %contains all pairs $(a,\phi)$ such that $\phi$ is defined according to the interpretation of sequential composition for the classical case:
%a pair $(a,A)$ is in $R\cdot S$ if and only if there exists ($\sum$) a (fuzzy) set $\phi_a$ that is related with $a$ and every state ($\prod$) in $\phi$ is related with some (fuzzy) set $\phi_b$.
Regarding the sequential composition, the expression for $\phi$ computes the weight of an execution that starts from a state $a$, arrives at a set of intermediate states $\phi_a$ and ends in a set of states $\varphi_b$.
The parallel composition considers the union of fuzzy sets for computing the external choice, which is just a generalisation of the set union  used for CPDL. 

% Next we present some operations over fuzzy multirelations.

%\subsection{Operations on fuzzy multirelations}

%\begin{definition}[Operations on fuzzy multirelations]\label{def:opfmultirel} Given a set $X$, a complete residuated lattice $\Lat$ over carrier $L$, a fuzzy set $\phi:X\to L$ and fuzzy multirelations $R$, $S$ over $X$,

	\subsection{Parametric construction of multi-valued concurrent dynamic logics}

	Each complete action lattice $\A$ induces a multi-valued, concurrent propositional dynamic logic $\CDL$, with weighted computations interpreted over $\A$. Its signature, formul\ae, semantics and satisfacton are presented below.\\

	\noindent \textbf{Signatures.} Signatures of $\CDL$ are pairs $(\Pi,\Prop)$ corresponding to the denotations of atomic programs and  propositions, respectively.

	\medskip
	\noindent \textbf{Formul\ae.}
	The \emph{set of composed programs}, denoted by $\Prog(\Pi)$, consists of all expressions generated by 
	$\pi \ni \pi_0 \,|\, \pi;\pi \,|\, \pi \cap \pi \,|\, \pi + \pi \,|\, \pi^*$, 
	for $\pi_0\in \Pi$. 
	Given a signature $(\Pi,\Prop)$, the $\CDL$-formul\ae\, for $(\Pi,\Prop)$, denoted by \\$\Fm(\Pi,\Prop)$, are the ones generated by the grammar\\
	$\rho \ni \top\,|\, \bot\,|\, p \,|\, \rho \vee \rho \,|\,\rho \wedge \rho \,|\, \rho \rightarrow \rho \,|\, \rho \leftrightarrow \rho \,|\,\langle \pi \rangle\rho$, %|\],[\pi]\rho\]
	for $p\in \Prop$ and $\pi \in \Prog(\Pi)$.

\medskip
	\noindent \textbf{Semantics.} The space where the programs are interpreted is given by the set of all fuzzy multirelations over a set of states $W$, denoted by $M(W)$, and the operations over its elements, as in Definition \ref{def:fuzzymrel}.

	$\CDL$-models for a signature $(\Pi,\Prop)$ are tuples $M=(W,V,\prgint{-})$ where $W$ is a set of states, $V$ is a valuation function $V:\Prop\times W\to A$ and $\prgint{-}$ attributes a fuzzy multirelation $R\subseteq W\times A^W$ to each atomic program.

	The interpretation of a program $\pi\in Prg(\Pi)$ in a model $M$ is recursively defined as:

 $\prgint{\pi;\pi'}=\prgint{\pi} \,\mathbf{\cdot} \,\prgint{\pi'},
	\prgint{\pi\cap\pi'}=\prgint{\pi} \,\mathbf{\cap} \,\prgint{\pi'}, 
	 \prgint{\pi + \pi'}=\prgint{\pi} \,\mathbf{\cup}\, \prgint{\pi'}\; \text{and}\; 
	\prgint{\pi^*}=\prgint{\pi}^{\textbf{*}}$.

	The satisfaction relation for a model $M=(W,V,\prgint{-})$
 is given by the valuation function $\cmodels:W\times \Fm(\Pi,\Prop) \rightarrow A$
 recursively defined as: 
	\begin{itemize}
		\item $(w\cmodels \top) = \top$
		\item $(w\cmodels \bot) = \bot$
		\item $(w\cmodels p) = V(p,w)$, for any $p\in \Prop$
		\item $(w\cmodels \rho \wedge \rho') = (w\cmodels \rho) \cdot (w\cmodels \rho')$
		\item $(w\cmodels \rho \vee \rho') = (w\cmodels \rho) + (w\cmodels \rho')$
		\item $(w\cmodels \rho \rightarrow \rho') = (w\cmodels \rho) \rightarrow (w\cmodels \rho')$
		\item $(w\cmodels \rho \leftrightarrow \rho') = (w\cmodels \rho \rightarrow \rho');(w\cmodels \rho' \rightarrow \rho)$
		\item $(w\cmodels \langle \pi \rangle\rho) = \sum\limits_{\phi\mid(w,\phi)\in\prgint{\pi}}\Big(\prod\limits_{u\in U}\big(\phi(u);(u\cmodels \rho)\big)\Big)$
		\item $(w\cmodels [\pi] \rho) = \prod\limits_{\phi\mid(w,\phi)\in\prgint{\pi}}\Big(\prod\limits_{u\in U}\big(\phi(u)\rightarrow(u\cmodels \rho)\big)\Big)$
	\end{itemize}

	\noindent where $U\subseteq W$.
 We say that $\rho$ is \emph{valid} when, for any model $M$, and for each state $w\in W$, $(w\cmodels\rho) =\top$.

The satisfaction of formula $(w\cmodels\langle\pi\rangle\rho)$ is given by the weight of some fuzzy set $\phi$ which is related with state $w$ by some fuzzy multirelation, and that of $\rho$ for every state of the domain of $\phi$. % Such interpretation is described by the definition presented in \cite{ConcurrentDL} for the diamond operator
 Moreover, the satisfaction for the box operator follows \cite{Madeira2016}, where every execution of the program must lead to a set of states all of which satisfy $\rho$.

As mentioned in Section \ref{sec:prelim}, the axiomatisation of CPDL was presented as being that of PDL, except for one that is restricted to atomic programs, plus an additional axiom for concurrency. Below we study such axiomatisation in the new models presented for $\CDL$.

%Now the validity of some axioms for the generated logics is studied.
First, Lemma \ref{lem:lemma3} provides some auxiliary properties used to prove next lemma.

\begin{lemma}\label{lem:lemma3}
Let $\AA$ be a complete $\DD$-action lattice. Then
	\begin{description}
		\item[(1.1)] $(w\cmodels \rho \rightarrow \rho')=\top$ iff $(w\cmodels\rho)\leq(w\cmodels\rho')$
		\item[(1.2)] $(w\cmodels \rho \leftrightarrow\rho')=\top$ iff  $(w\cmodels\rho)=(w\cmodels\rho')$
	\end{description}	
\end{lemma}

\begin{proof}
Analogous to \cite{Madeira2016}.
\medskip
\hfill $\square$
\end{proof}

% The following results state the validity of some axioms of CPDL.

\begin{lemma}\label{lem:lemma4}
Let $\A$ be a a complete $\DD$-action lattice. The following are valid formul\ae\, in any $\CDL$:
	\begin{description}

		\item[(2.1)] $\langle \pi_0\rangle (\rho \vee \rho') \leftrightarrow \langle \pi_0\rangle \rho \vee \langle \pi_0\rangle \rho'$\
		\item[(2.2)] $\langle \pi\rangle (\rho \wedge \rho') \rightarrow \langle \pi\rangle \rho \wedge \langle \pi\rangle \rho'$
		\item[(2.3)] $\langle \pi + \pi'\rangle \rho \leftrightarrow \langle \pi\rangle \rho \vee \langle \pi\rangle \rho$
		\item[(2.4)] $\langle \pi\rangle \bot \leftrightarrow \bot$
		\item[(2.5)] $\langle \pi\cap\pi'\rangle\rho\leftrightarrow \langle\pi\rangle\rho\wedge\langle\pi'\rangle\rho$
		\item[(2.6)] $[\pi+\pi']\rho\leftrightarrow[\pi]\rho\wedge[\pi']\rho$
		\item[(2.7)] $[\pi](\rho\wedge\rho')\rightarrow[\pi]\rho\wedge[\pi]\rho'$
	
	\end{description}

\end{lemma}

\begin{proof}
	The proof uses the satisfaction function $\cmodels$ and some axioms and properties of action lattices. The technical details are documented in Appendix \ref{sec:appendix}.
	\medskip
	\hfill $\square$
	
\end{proof}

\section{Conclusion}\label{sec:conc}

We took, in this paper, the first step in order to develop a rigorous and systematic formalism for the verification of weighted concurrent systems, motivated by the fuzzy case. The approach is based on the combination of some ideas from previous research \cite{DBLP:conf/sbmf/MadeiraNMB14,ConcurrentDL,SCP} to characterise both the computational and logical settings on top of which a proper (axiomatic, denotational and operational) semantics for fuzzy programs will be developed, in future work.

%The attentive reader must have noticed the absence of $[\_]$-modalities, which is typically part of PDL, in the grammar of the logic. Since the presented work consists on the first step to generalise CPDL, we aim at studying additional definitions in future work. Such is the case also for the soundness of the logic: the validity of axioms involving the $[\_]$-modality, as well as operators $;$ and $^*$ will be studied in the future.

There are numerous research lines that were left open and are worth to pursue in the near future. The most obvious is the study of a proper complete axiomatisation for the generated logics. In particular, the validity of the remaining axioms of CPDL, namely the ones involving operators $;$ and $^*$, will be analysed in the new models. Another relevant path to be followed would be to study the relations between PDL, CPDL and their graded variants. In one direction, we propose to investigate whether CPDL can be obtained from $\CDL$ by taking $\bool$ as lattice. Other would be to study if there is a way to obtain multi-valued PDL as special case of $\CDL$, such that there is a correspondence between the operations for fuzzy multirelations and operations on matrices. Additionally, relevant results about decidability and complexity of the logics are naturally in our agenda.

%Beyond this scenario, binary multirelations are versatile enough to be used in other contexts, such as to formalise logics for games of choice \cite{MartinCR07}, or as an alternative to predicate transformers for reasoning about programs \cite{Rewitzky03}.

Although we base our definition of sequential composition for fuzzy multirelations in that of Peleg, there are other versions of the operator worth to be analysed. One corresponds to the definition introduced for giving semantics to Parikh's game logic \cite{DBLP:conf/focs/Parikh83}

$$R\cdot S=\Big\{(a,A)\mid \exists B.(a,B)\in R\wedge \exists f.(\forall b\in B. (b,A)\in S)\}$$

\noindent It is clearly stronger than Peleg's, since it requires that every intermediate state $b$ must be related with the arriving set of states $A$. Another one, the Kleisli composition, was later studied in \cite{DBLP:journals/jlp/FurusawaKST17}. It is our goal to introduce proper generalisations of such operations, with possible applications in scenarios like a graded variant of game logics, as well as the development of axiomatic systems for each variation.

Finally, we propose to adapt the models of the generated logics in order to allow the introduction of assignments of variables to values in a given data domain. The goal is to develop (parametric) logics for the verification of programs written in a fuzzy imperative programming language, such as \cite{DBLP:journals/artmed/VetterleinMA10} or \cite{DBLP:journals/ijcisys/CingolaniA13}.

%\textcolor{red}{Moreover, it would be nice to see a formal proof that the proposed logic is a proper extension of the underlying logics, i.e. that CDPL and multi-valued PDL are special cases of the new logic. I see that for getting CDPL one must take as lattice the Boolean lattice 2. But how can we get multi-valued PDL as special case such that the operations for fuzzy binary relations correspond to operations on matrices?}

%An interesting case is the analysis of axiom \ref{eq:ax2}, and how the restriction to atomic programs could be handled in the semantics of fuzzy multirelations.

% Instead of focusing on a specific lattice to characerise the computaional model and the truth space (for instance the interval $[0,1]$ with appropriate operators) we opt by following previous work \cite{Madeira2016} to present a generic, parametric method for generating multi-valued CDL.

% \TODO{1. Referir que este trabalho surge no contexto {}do desenvolvimento de uma tese de doutoramento que se suporta em generalizações de lógicas dinamicas, sendo isto uma primeira tentativa the generalização da lógica do Peleg.\\	
% 2. Discussão sobre o axioma (2), no que toca à restrição para programas atómicos (estudo futuro sobre o porquê desta restrição)\\
% 3. apresentação da definição de uma modalidade apenas (deixando o estudo do box para trabalho futuro\\)
% 4. apresentação apenas de alguns axiomas da PDL (trabalho futuro para os outros, incluindo os axiomas que envolvem o box e o star)}

\bibliographystyle{acm}
 \bibliography{thesis}

\begin{thebibliography}{10}

\bibitem{Baltag2012}
{\sc Baltag, A., and Smets, S.}
\newblock {The dynamic turn in quantum logic}.
\newblock {\em Synthese 186}, 3 (2012), 753--773.

\bibitem{DBLP:journals/ijcisys/CingolaniA13}
{\sc Cingolani, P., and Alcal{\'{a}}{-}Fdez, J.}
\newblock jfuzzylogic: a java library to design fuzzy logic controllers
  according to the standard for fuzzy control programming.
\newblock {\em Int. J. Comput. Intell. Syst. 6}, sup1 (2013), 61--75.

\bibitem{conway1971}
{\sc Conway, J.}
\newblock {\em {Regular Algebra and Finite Machines}}.
\newblock Dover Publications, 1971.

\bibitem{probProgUsingHoareL}
{\sc den Hartog, J., and de~Vink, E.~P.}
\newblock Verifying probabilistic programs using a {H}oare like logic.
\newblock {\em Int. J. Found. Comput. Sci. 13}, 3 (2002), 315--340.

\bibitem{DenHartog2002}
{\sc den Hartog, J.~I.}
\newblock {\em {Probabilistic Extensions of Semantical Models}}.
\newblock PhD thesis, {Vrije Universiteit}, Vrije, 2002.

\bibitem{DBLP:journals/jlp/FurusawaKST17}
{\sc Furusawa, H., Kawahara, Y., Struth, G., and Tsumagari, N.}
\newblock Kleisli, {P}arikh and {P}eleg compositions and liftings for
  multirelations.
\newblock {\em J. Log. Algebr. Meth. Program. 90\/} (2017), 84--101.

\bibitem{DBLP:journals/tocl/FurusawaS16}
{\sc Furusawa, H., and Struth, G.}
\newblock Taming multirelations.
\newblock {\em {ACM} Trans. Comput. Log. 17}, 4 (2016), 28:1--28:34.

\bibitem{SCP}
{\sc Gomes, L., Madeira, A., and Barbosa, L.~S.}
\newblock {Generalising KAT to verify weighted computations}.
\newblock Tech. rep., HASLab INESC TEC - Univ. of Minho, Portugal, Department
  of Informatics, 2018.

\bibitem{DBLP:conf/concur/HoareMSW09}
{\sc Hoare, C. A.~R., M{\"{o}}ller, B., Struth, G., and Wehrman, I.}
\newblock Concurrent {K}leene algebra.
\newblock In {\em {CONCUR} 2009. Proceedings\/} (2009), M.~Bravetti and
  G.~Zavattaro, Eds., vol.~5710 of {\em LNCS}, Springer, pp.~399--414.

\bibitem{DBLP:journals/jlp/JipsenM16}
{\sc Jipsen, P., and Moshier, M.~A.}
\newblock Concurrent {K}leene algebra with tests and branching automata.
\newblock {\em J. Log. Algebr. Meth. Program. 85}, 4 (2016), 637--652.

\bibitem{Kozen1985}
{\sc Kozen, D.}
\newblock {A probabilistic PDL}, 1985.

\bibitem{onactionalgebras}
{\sc Kozen, D.}
\newblock On action algebras.
\newblock In {\em Logic and the Flow of Information\/} (Amsterdam, 1993).

\bibitem{kleenematrices}
{\sc Kozen, D.}
\newblock {A Completeness Theorem for Kleene Algebras and the Algebra of
  Regular Events}.
\newblock {\em Information and Computation 110}, May 1994 (1994), 366--390.

\bibitem{Madeira2016}
{\sc Madeira, A., Neves, R., and Martins, M.~A.}
\newblock {An exercise on the generation of many-valued dynamic logics}.
\newblock {\em JLAMP 1\/} (2016), 1--29.

\bibitem{DBLP:conf/sbmf/MadeiraNMB14}
{\sc Madeira, A., Neves, R., Martins, M.~A., and Barbosa, L.~S.}
\newblock A dynamic logic for every season.
\newblock In {\em Formal Methods: Foundations and Applications. Proceedings\/}
  (2014), C.~Braga and N.~Mart{\'{\i}}{-}Oliet, Eds., vol.~8941 of {\em LNCS},
  Springer, pp.~130--145.

\bibitem{DBLP:journals/corr/McIverRS13}
{\sc McIver, A., Rabehaja, T.~M., and Struth, G.}
\newblock Probabilistic concurrent {K}leene algebra.
\newblock In {\em Proceedings {QAPL}.\/} (2013), L.~Bortolussi and H.~Wiklicky,
  Eds., vol.~117 of {\em {EPTCS}}, pp.~97--115.

\bibitem{DBLP:conf/focs/Parikh83}
{\sc Parikh, R.}
\newblock Propositional game logic.
\newblock In {\em 24th Annual Symposium on Foundations of Computer Science\/}
  (1983), {IEEE} Computer Society, pp.~195--200.

\bibitem{ConcurrentDL}
{\sc Peleg, D.}
\newblock Concurrent dynamic logic.
\newblock {\em J. {ACM} 34}, 2 (1987), 450--479.

\bibitem{Platzer}
{\sc Platzer, A.}
\newblock {\em Logical Analysis of Hybrid Systems - Proving Theorems for
  Complex Dynamics}.
\newblock Springer, 2010.

\bibitem{Qiao2008}
{\sc Qiao, R., Wu, J., Wang, Y., and Gao, X.}
\newblock {Operational semantics of probabilistic Kleene algebra with tests}.
\newblock {\em Proceedings - IEEE Symposium on Computers and Communications\/}
  (2008), 706--713.

\bibitem{DBLP:journals/artmed/VetterleinMA10}
{\sc Vetterlein, T., Mandl, H., and Adlassnig, K.}
\newblock Fuzzy arden syntax: {A} fuzzy programming language for medicine.
\newblock {\em Artificial Intelligence in Medicine 49}, 1 (2010), 1--10.

\bibitem{ZADEH1965338}
{\sc Zadeh, L.}
\newblock Fuzzy sets.
\newblock {\em Information and Control 8}, 3 (1965), 338 -- 353.

\end{thebibliography}

 \newpage
 \appendix
 %!TEX root = main.tex

\section{Appendix}\label{sec:appendix}

The proofs here presented resort to some axioms of action lattices, enumerated below, and properties stated by Lemma \ref{lem:lemma1}.

\begin{eqnarray*}
				a+ (b +c) & = & (a+ b) +c \label{eqn1}\\
				 a +b & = & b + a \label{eqn2}\\
				 a + 0 & = &  0 + a  = a\label{eqn4}\\
  				a;(b +c ) & =& (a;b) + (a;c) \label{eqn7}\\
				a;0 & = & 0;a = 0 \label{eqn9}
\end{eqnarray*}

 \begin{lemma}\label{lem:lemma1}
 The following properties hold for any action lattice $\A$:

 \begin{eqnarray}
 a \leq b \; \& \; c\leq d & \Rightarrow & a+c \leq b+d \label{eq:eq33} \\
 	a;(b\cdot c) & \leq & (a;b)\cdot (a;c)	 \label{eq:eq36}
 \end{eqnarray}

 \noindent For $I$ finite, we also have

 \begin{equation}
 	\sum_{i\in I}(a_i\cdot b_i) \leq  \sum_{i\in I} a_i \cdot \sum_{i\in I} b_i
 \label{eq:eq45}
 \end{equation}

\end{lemma}

\noindent \textbf{Proof of Lemma 2}\label{app:lemma2}\\

\textbf{(2.1)}:

\begin{eqnarray*}
	  \arrayin{
	  	& & w \cmodels \langle \pi \rangle (\rho \vee \rho'))
	 	\just={definition of $\cmodels$}
	 	\sum\limits_{\phi\mid(w,\phi)\in\prgint{\pi}}\Big(\prod\limits_{u\in U}\big(\phi(u); (u\cmodels \rho \vee \rho')\big)\Big)
	 	\just={definition of $\cmodels$}  
	 	\sum\limits_{\phi\mid(w,\phi)\in\prgint{\pi}}\Big(\prod\limits_{u\in W}\big((\phi(u); \\& & 
	 	\big((u\cmodels \rho) + (u\cmodels \rho')\big)\big)\Big)
	 	\just={(\ref{eqn7})}
	 	\sum\limits_{\phi\mid(w,\phi)\in\prgint{\pi}}\Big(\prod\limits_{w'\in W}\big(\phi(w'); (w'\cmodels \rho) + \\ & &\phi(w'); (w'\cmodels \rho')\big)\Big)
	  }
	  \vline \arrayin{ & &
		 \justdois={$\pi_0$ is atomic hence}{$w$ is related with a singleton $\{u\}$}
		 \sum\limits_{\phi\mid(w,\phi)\in\prgint{\pi}}\big(\phi(w'); (w'\cmodels \rho) + \\ & &\phi(w'); (w'\cmodels \rho')\big)
	  \just={by (\ref{eqn1}) and (\ref{eqn2})}
	 	{\scriptstyle
	 	\sum\limits_{\phi\mid(w,\phi)\in\prgint{\pi}}\Big(\sum\limits_{w'\in W}\big(\phi(w'); (w'\cmodels \rho)\big)\Big)} \\ + & & {\scriptstyle\sum_{\phi\mid(w,\phi)\in\prgint{\pi}}\Big(\sum_{w'\in W}\big(\phi(w'); (w'\cmodels \rho')\big)\Big)}
	 	\just={definition of $\cmodels$}
	 	(w\cmodels \langle \pi\rangle \rho) + (w\cmodels \langle \pi\rangle \rho')
	 	\just={definition of $\cmodels$}
	 	(w\cmodels \langle \pi\rangle \rho \vee  \langle \pi\rangle \rho')
	 	}
	 	\end{eqnarray*}

	 	\noindent Therefore, by Lemma \ref{lem:lemma3}, $\langle \pi \rangle (\rho \vee\rho') \leftrightarrow \langle \pi\rangle \rho \vee\langle \pi\rangle \rho$ is valid.
		 \medskip

\textbf{(2.2)}:

\begin{eqnarray*}
\arrayin{
	 	& & (w \cmodels \langle \pi \rangle (\rho \wedge \rho'))
	 	\just={definition of $\cmodels$}
	 	\sum\limits_{\phi\mid(w,\phi)\in\prgint{\pi}}\Big(\sum\limits_{w'\in W}\big(\phi(w'); (w'\cmodels \rho \wedge \rho')\big)\Big)
	 	\just={definition of $\cmodels$}  
	 	\sum\limits_{\phi\mid(w,\phi)\in\prgint{\pi}}\Big(\sum\limits_{w'\in W}\big(\phi(w'); \\ & &  \big((w'\cmodels \rho) \cdot (w'\cmodels \rho')\big)\big)\Big)
	 	\just\leq{by (\ref{eq:eq36}) and (\ref{eq:eq33})}
	 	\sum\limits_{\phi\mid(w,\phi)\in\prgint{\pi}}\Big(\sum\limits_{w'\in W}\big((\phi(w'); (w'\cmodels \rho)) \cdot \\ & &  (\phi(w'); (w'\cmodels \rho'))\big)\Big)
	 	} \vline \arrayin{ & &
	 	\just\leq{by (\ref{eq:eq45}) }
	 	{\scriptstyle \sum\limits_{\phi\mid(w,\phi)\in\prgint{\pi}}\Big(\sum\limits_{w'\in W}\big(\phi(w'); (w'\cmodels \rho)\big)\Big)} \\ \cdot & & {\scriptstyle \sum_{\phi\mid(w,\phi)\in\prgint{\pi}}\Big(\sum_{w'\in W}\big(\phi(w'); (w'\cmodels \rho'))\big)\Big)}
	 	\just={definition of $\cmodels$}
	 	(w\cmodels \langle \pi\rangle \rho) \cdot (w\cmodels \langle \pi\rangle \rho')
	 	\just={definition of $\cmodels$}
	 	(w\cmodels \langle \pi\rangle \rho \wedge  \langle \pi\rangle \rho')}
	 \end{eqnarray*}
Therefore, by Lemma \ref{lem:lemma3}, $\langle \pi \rangle (\rho \wedge \rho')\rightarrow \langle \pi\rangle \rho \wedge\langle \pi\rangle \rho'$ is valid.
\medskip

%\section{Proof of Lemma 5}\label{app:lemma5}

\textbf{(2.3)}:
	  \begin{eqnarray*}
	  \arrayin{
	  	& & (w \cmodels \langle \pi + \pi'\rangle \rho)
	  	\just={definition of $\cmodels$}
	  	\sum\limits_{\phi\mid(w,\phi)\in\prgint{\pi+\pi'}}\Big(\sum\limits_{w'\in W}\big(\phi(w'); (w'\cmodels \rho)\big)\Big)
	 	 \just={definition of $\prgint{\pi+\pi'}$}  
	 	 \sum\limits_{\substack{\phi\mid(w,\phi)\in\prgint{\pi}\\
	 	 \vee \phi\mid(w,\phi)\in\prgint{\pi'}}}\Big(\sum\limits_{w'\in W}\big(\phi(w'); (w'\cmodels \rho)\big)\Big)
	 	 %
	 	 % \just={distributivity (7)}
	 	 %
	 	 }\vline \arrayin{ & &
	 	 \just={definition of $\prgint{\pi+\pi'}$}  
	 	 \sum\limits_{\phi\mid(w,\phi)\in\prgint{\pi}}\Big(\sum\limits_{w'\in W}\big(\phi(w'); (w'\models \rho)\big)\Big)\\ + & & \sum\limits_{\phi\mid(w,\phi)\in\prgint{\pi'}}\Big(\sum\limits_{w'\in W}\big(\phi(w'); (w'\models \rho)\big)\Big)
	 	 \just={definition of $\cmodels$}
	 	 (w\models \langle \pi\rangle \rho) + (w\models \langle \pi'\rangle \rho)
	 	 \just={definition of $\cmodels$}
	 	 (w\models \langle \pi\rangle \rho \vee \langle \pi'\rangle \rho)
	 	 }
	  \end{eqnarray*}
	  Therefore, by Lemma \ref{lem:lemma3}, $\langle \pi + \pi'\rangle \rho \leftrightarrow \langle \pi\rangle \rho \vee  \langle \pi'\rangle \rho$ is valid.
	  \medskip

\textbf{(2.4)}:
\begin{eqnarray*}
		\arrayin{
	 	& & (w\cmodels \langle \pi\rangle \bot)
	 	\just={definition of $\cmodels$}
	 	\sum\limits_{\phi\mid(w,\phi)\in\prgint{\pi}}\Big(\sum\limits_{w'\in W}\big(\phi(w');(w'\cmodels \bot)\big)\Big)
	 	\just={definition of satisfaction}
		\sum\limits_{\phi\mid(w,\phi)\in\prgint{\pi}}\Big(\sum\limits_{w'\in W}\big(\phi(w');\bot\big)\Big)
		}\vline\arrayin{
	 	& & \just={by (\ref{eqn9})}
		\sum\limits_{\phi\mid(w,\phi)\in\prgint{\pi}}\Big(\sum\limits_{w'\in W}\bot\Big)
	 	\just={by (\ref{eqn4})}
		\bot
		}
	 \end{eqnarray*}
Therefore, by Lemma \ref{lem:lemma3}, $\langle \pi\rangle \bot \leftrightarrow \bot$ is valid.
\medskip

%\section{Proof of Lemma 6}\label{app:lemma6}

\textbf{(2.5)}:

\begin{eqnarray*}
	\arrayin{
		& & w\cmodels \langle \pi\cap\pi'\rangle\rho
		\just={definition of $\cmodels$}
		\sum\limits_{\phi\mid(w,\phi)\in\prgint{\pi\cap\pi'}}\Big(\sum\limits_{w'\in W}\big(\phi(w');w'\cmodels\rho\big)\Big)
		\just={definition of $\prgint{\pi\cap\pi'}$}
		\sum\limits_{\phi_1\mid (w,\phi_1)\in \prgint{\pi}}\Big(\sum\limits_{w'\in W}\big(\phi_1(w');w'\cmodels\rho\big)\Big)\\ & &\cdot \sum\limits_{\phi_2\mid(w,\phi_2)\in \prgint{\pi'}}\Big(\sum_{w'\in W}\big(\phi_2(w');w'\cmodels\rho\big)\Big)	
		\just={definition of $\cmodels$}
		(w\cmodels\langle\pi\rangle\rho)\cdot (w\cmodels\langle\pi'\rangle\rho)
		\just={definition of $\cmodels$}
	(w\cmodels\langle\pi\rangle\rho\wedge \langle\pi'\rangle\rho)
	}
	\end{eqnarray*}
	Therefore, by Lemma \ref{lem:lemma3}, $\langle \pi\cap\pi'\rangle\rho\leftrightarrow \langle\pi\rangle\rho\wedge\langle\pi'\rangle\rho$
	 is valid.
	\medskip

%	\section{Proof of Lemma 7}\label{app:lemma7}
 
\textbf{(2.6)}:

	\begin{eqnarray*}
		\arrayin{
& & w\cmodels [\pi+\pi']\rho
\just={definition of $\cmodels$}
\prod\limits_{\phi\mid(w,\phi)\in\prgint{\pi+\pi'}}\Big(\prod\limits_{w'\in W}\big(\phi(w')\to w'\cmodels\rho\big)\Big)
\just={definition of $\prgint{\pi+\pi'}$}
\prod\limits_{\substack{\phi\mid(w,\phi)\in\prgint{\pi}\\
\vee \phi\mid(w,\phi)\in\prgint{\pi'}}}\Big(\prod\limits_{w'\in W}\big(\phi(w')\to w'\cmodels\rho\big)\Big)
		\justdois={$\prod$ iterates over all $\phi$ such that}{$(w,\phi)\in \prgint{\pi} \vee (w,\phi)\in \prgint{\pi'}$}
		\begin{aligned}
			\prod\limits_{\phi\mid(w,\phi)\in\prgint{\pi}}\Big(\prod\limits_{w'\in W}\big(\phi(u)\to w'\cmodels\rho\big)\Big)\\ \cdot
		\prod\limits_{\phi\mid(w,\phi)\in\prgint{\pi'}}\Big(\prod\limits_{w'\in W}\big(\phi(u)\to w'\cmodels\rho\big)\Big)
		\end{aligned}
		\just={definition of $\cmodels$}
		(w\cmodels[\pi]\rho)\cdot(w\cmodels[\pi']\rho)
		\just={definition of $\cmodels$}
		(w\cmodels[\pi]\rho\wedge[\pi']\rho)
}
	\end{eqnarray*}	
\medskip

\textbf{(2.7)}:

	\begin{eqnarray*}
		\arrayin{
		& &	[\pi](\rho\wedge\rho')
		\just={definition of $\cmodels$}
		\prod\limits_{\phi\mid(w,\phi)\in\prgint{\pi}}\Big(\prod\limits_{w'\in W}\big(\phi(w')\to w'\cmodels\rho\wedge\rho'\big)\Big)
		\just={definition of $\cmodels$}
		\prod\limits_{\phi\mid(w,\phi)\in\prgint{\pi}}\Big(\prod\limits_{w'\in W}\big(\phi(w')\to (w'\cmodels\rho\cdot w\cmodels\rho')\big)\Big)
		\just\leq{definition of $\cmodels$}
		\prod\limits_{\phi\mid(w,\phi)\in\prgint{\pi}}\Big(\prod\limits_{w'\in W}\big((\phi(w')\to w'\cmodels\rho)\cdot (\phi(w')\to w'\cmodels\rho')\big)\Big)
		\just={(\ref{eqn1})}
		\begin{aligned}
			\prod\limits_{\phi\mid(w,\phi)\in\prgint{\pi}}\Big(\prod\limits_{w'\in W}\big(\phi(w')\to w'\cmodels\rho\big)\Big)\\ \cdot
		\prod\limits_{\phi\mid(w,\phi)\in\prgint{\pi}}\Big(\prod\limits_{w'\in W}\big(\phi(w')\to w'\cmodels\rho'\big)\Big)
		\end{aligned}
		\just={definition of $\cmodels$}
		(w\cmodels[\pi]\rho)\cdot(w\cmodels[\pi]\rho')
		\just={definition of $\cmodels$}
		(w\cmodels[\pi]\rho\wedge[\pi]\rho')
		}
	\end{eqnarray*}

 \end{document}